\PassOptionsToPackage{unicode}{hyperref}
\PassOptionsToPackage{hyphens}{url}
\PassOptionsToPackage{dvipsnames,svgnames,x11names}{xcolor}
\documentclass[
  12pt]{article}
\usepackage{bbm}
\usepackage{algorithm}
\usepackage{algorithmicx}
\usepackage{algpseudocode}

\usepackage{amsmath,amssymb,amsfonts}

\newtheorem{theorem}{Theorem}

\newtheorem{remark}{Remark}

\newtheorem{corollary}{Corollary}
\newtheorem{assumption}{Assumption}[section]
\newtheorem{proof}{Proof}  

\usepackage{iftex}
\ifPDFTeX
  \usepackage[T1]{fontenc}
  \usepackage[utf8]{inputenc}
  \usepackage{textcomp} 
\else 
  \usepackage{unicode-math}
  \defaultfontfeatures{Scale=MatchLowercase}
  \defaultfontfeatures[\rmfamily]{Ligatures=TeX,Scale=1}
\fi
\usepackage{lmodern}
\ifPDFTeX\else  
\fi
\IfFileExists{upquote.sty}{\usepackage{upquote}}{}
\IfFileExists{microtype.sty}{
  \usepackage[]{microtype}
  \UseMicrotypeSet[protrusion]{basicmath} 
}{}
\makeatletter
\@ifundefined{KOMAClassName}{
  \IfFileExists{parskip.sty}{%
    \usepackage{parskip}
  }{
    \setlength{\parindent}{0pt}
    \setlength{\parskip}{6pt plus 2pt minus 1pt}}
}{
  \KOMAoptions{parskip=half}}
\makeatother
\usepackage{xcolor}
\makeatletter
\ifx\paragraph\undefined\else
  \let\oldparagraph\paragraph
  \renewcommand{\paragraph}{
    \@ifstar
      \xxxParagraphStar
      \xxxParagraphNoStar
  }
  \newcommand{\xxxParagraphStar}[1]{\oldparagraph*{#1}\mbox{}}
  \newcommand{\xxxParagraphNoStar}[1]{\oldparagraph{#1}\mbox{}}
\fi
\ifx\subparagraph\undefined\else
  \let\oldsubparagraph\subparagraph
  \renewcommand{\subparagraph}{
    \@ifstar
      \xxxSubParagraphStar
      \xxxSubParagraphNoStar
  }
  \newcommand{\xxxSubParagraphStar}[1]{\oldsubparagraph*{#1}\mbox{}}
  \newcommand{\xxxSubParagraphNoStar}[1]{\oldsubparagraph{#1}\mbox{}}
\fi
\makeatother

\usepackage{longtable,booktabs,array}
\usepackage{calc} 
\usepackage{etoolbox}
\makeatletter
\patchcmd\longtable{\par}{\if@noskipsec\mbox{}\fi\par}{}{}
\makeatother
\IfFileExists{footnotehyper.sty}{\usepackage{footnotehyper}}{\usepackage{footnote}}
\makesavenoteenv{longtable}
\usepackage{graphicx}
\makeatletter
\def\maxwidth{\ifdim\Gin@nat@width>\linewidth\linewidth\else\Gin@nat@width\fi}
\def\maxheight{\ifdim\Gin@nat@height>\textheight\textheight\else\Gin@nat@height\fi}
\makeatother
\setkeys{Gin}{width=\maxwidth,height=\maxheight,keepaspectratio}
\makeatletter
\def\fps@figure{htbp}
\makeatother

\makeatletter
\@ifpackageloaded{caption}{}{\usepackage{caption}}
\AtBeginDocument{%
\ifdefined\contentsname
  \renewcommand*\contentsname{Table of contents}
\else
  \newcommand\contentsname{Table of contents}
\fi
\ifdefined\listfigurename
  \renewcommand*\listfigurename{List of Figures}
\else
  \newcommand\listfigurename{List of Figures}
\fi
\ifdefined\listtablename
  \renewcommand*\listtablename{List of Tables}
\else
  \newcommand\listtablename{List of Tables}
\fi
\ifdefined\figurename
  \renewcommand*\figurename{Figure}
\else
  \newcommand\figurename{Figure}
\fi
\ifdefined\tablename
  \renewcommand*\tablename{Table}
\else
  \newcommand\tablename{Table}
\fi
}
\@ifpackageloaded{float}{}{\usepackage{float}}
\floatstyle{ruled}
\@ifundefined{c@chapter}{\newfloat{codelisting}{h}{lop}}{\newfloat{codelisting}{h}{lop}[chapter]}
\floatname{codelisting}{Listing}

\makeatother
\makeatletter
\@ifpackageloaded{caption}{}{\usepackage{caption}}
\@ifpackageloaded{subcaption}{}{\usepackage{subcaption}}
\makeatother

\ifLuaTeX
  \usepackage{selnolig}  
\fi
\usepackage[]{natbib}
\usepackage{bookmark}

\IfFileExists{xurl.sty}{\usepackage{xurl}}{} 
\hypersetup{
  pdftitle={Title},
  pdfauthor={Author 1; Author 2},
  pdfkeywords={3 to 6 keywords, that do not appear in the title},
  colorlinks=true,
  linkcolor={blue},
  filecolor={Maroon},
  citecolor={Blue},
  urlcolor={Blue},
  pdfcreator={LaTeX via pandoc}}

\newcommand{\anon}{1}

\begin{document}

\def\spacingset#1{\renewcommand{\baselinestretch}%
{#1}\small\normalsize} \spacingset{1}


\if1\anon
{
  \title{\bf MMDCP: A Distribution-free Approach to Outlier Detection and Classification with Coverage Guarantees and SCW-FDR Control}
  
    \author{\small
    	Youwu Lin\textsuperscript{1,2}, 
    	Xiaoyu Qian\textsuperscript{2,3}, 
    	Jinru Wu\textsuperscript{2}, 
    	Qi Liu\textsuperscript{2}, 
    	Pei Wang\textsuperscript{4}\thanks{Corresponding author, Email: pei.wang@gdufs.edu.cn}\\[6pt]
    	\small \textsuperscript{1} Guanghua School of Management, Peking University, 100871, Beijing, P.R.China\\
    	\small \textsuperscript{2} School of Mathematics and Computing Science, \\
    	\small Guilin University of Electronic Technology, 541002, Guilin, P.R.China\\
    	\small \textsuperscript{3} College of Energy and Power Engineering, \\ \small Nanjing University of Aeronautics and Astronautics, 210016, Nanjing, P.R.China\\
    	\small \textsuperscript{4} School of Business, Guangdong University of Foreign Studies, 510006, Guangzhou, P.R.China
    }

  \maketitle
} \fi

\if0\anon
{
  \bigskip
  \bigskip
  \bigskip
  \begin{center}
    {\LARGE\bf Title}
\end{center}
  \medskip
} \fi

\bigskip
\begin{abstract}
We propose the Modified Mahalanobis Distance Conformal Prediction (MMDCP), a unified framework for multi-class classification and outlier detection under label shift, where the training and test distributions may differ. In such settings, many existing methods construct nonconformity scores based on empirical cumulative or density functions combined with data-splitting strategies. However, these approaches are often computationally expensive due to their heavy reliance on resampling procedures and tend to produce overly conservative prediction sets with unstable coverage, especially in small samples. To address these challenges, MMDCP combines class-specific distance measures with full conformal prediction to construct a score function, thereby producing adaptive prediction sets that effectively capture both inlier and outlier structures. Under mild regularity conditions, we establish convergence rates for the resulting sets and provide the first theoretical characterization of the gap between oracle and empirical conformal $p$-values, which ensures valid coverage and effective control of the class-wise false discovery rate (CW-FDR). We further introduce the Summarized Class-Wise FDR (SCW-FDR), a novel global error metric aggregating false discoveries across classes, and show that it can be effectively controlled within the MMDCP framework. Extensive simulations and two real-data applications support our theoretical findings and demonstrate the advantages of the proposed method.

\end{abstract}

{\it Keywords:} Modified Mahalanobis distance, Full conformal inference, Label shift, Set-valued prediction, Multiple hypothesis testing
\vfill

\newpage
\spacingset{1.8} 

\section{Introduction}
We address the joint problem of multi-class classification and outlier detection using a conformal inference framework. This task arises in domains such as finance \citep{bao2020detecting, hilal2022financial}, medicine \citep{zhao2021anomaly, gu2023lightweight}, and network security \citep{mothukuri2021federated}. Classical statistical models rely on restrictive assumptions \citep{markou2003novelty}, whereas machine learning methods offer flexibility but they often lack uncertainty quantification and finite-sample guarantees \citep{sadinle2019least, liang2024integrative}. These shortcomings underscore the need for approaches that are both assumption-free and capable of delivering rigorous uncertainty control in multi-class settings with outliers.

Various set-valued prediction methods based on the conformal prediction framework \citep{vovk2005algorithmic, shafer2008tutorial} have been developed to address these limitations. The core idea of conformal prediction is to construct a nonconformity score function under the assumption of data exchangeability. The choice of score function is therefore critical, as it directly governs predictive performance \citep{chernozhukov2021distributional, angelopoulos2023conformal}. Numerous approaches have been proposed to obtain high-performance nonconformity scores. Machine learning–based methods leverage flexible models to define such scores \citep{bates2023testing, liang2024integrative}, but their performance can be unstable across datasets due to model sensitivity. Another line of work employs density-level set methods, which estimate nonconformity scores via density estimation \citep{cadre2006kernel, rigollet2009optimal, lei2013distribution, sadinle2019least}. However, these approaches suffer from the curse of dimensionality \citep{gu2013nonparametric} and ignore cross-class comparisons, often resulting in degraded classification performance \citep{guan2022prediction}.

To address the limitations of density-level set methods, \citet{guan2022prediction} proposed Balanced and Conformal Optimized Prediction Sets (BCOPS), which use a density ratio as the nonconformity score. However, BCOPS requires random sample splitting and binary classification for ratio estimation, resulting in quadratic computational growth with the number of classes \citep{wu2025class, ding2023class}. In addition, sample splitting not only reduces the efficiency of data use, which leads to lower predictive accuracy and wider prediction intervals \citep{linusson2017calibration}, but also induces instability across random splits, a phenomenon commonly referred to as the “$p$-value lottery” \citep{meinshausen2009p}. To alleviate this, \citet{han2024conformalized} developed the Conformalized Semi-Supervised Random Forest (CSForest) for simultaneous classification and outlier detection. Nevertheless, both BCOPS and CSForest rely on probability or density estimation, which is notoriously difficult in high dimensions \citep{wu2010robust,zhang2013effect}.

We propose a novel method, Modified Mahalanobis Distance Conformal Prediction (MMDCP), for simultaneous outlier detection and multi-class classification. MMDCP combines the modified Mahalanobis distance with full conformal prediction to construct an adaptive, model-free score function. Its model-free nature mitigates sensitivity to model specification, a limitation commonly encountered by machine learning–based approaches. In addition, by avoiding nonparametric density estimation, MMDCP overcomes the challenges such procedures face in high-dimensional settings. As a result, MMDCP delivers robust empirical performance across diverse datasets and scenarios. 

Given their capacity to capture geometric relationships among samples, many distance-based methods have been developed for classification and outlier detection, with seminal examples including Fisher discriminant analysis, the naive Bayes classifier \citep{bickel2004some, fan2008high}, and Mahalanobis distance–based approaches \citep{xiang2008learning}. However, most of these methods primarily focus on class assignment while paying limited attention to quantifying predictive uncertainty. To address this limitation, \citet{katsios2024multi} combined Mahalanobis distance with the split conformal prediction framework for multi-label classification, but they also highlighted that the Mahalanobis-based nonconformity score needs further investigation for reliable uncertainty quantification. While their work represents an important step toward integrating distance-based measures with conformal prediction, its scope remains limited.

Our study builds on this direction but differs from \citet{katsios2024multi} in three principal respects. First, we address multi-class classification under distributional shift, whereas their study focuses on multi-label classification. Second, we employ a conformal score function that is distinct from theirs. Third, and most importantly, we establish new theoretical results, including the first formal characterization of the gap between oracle and empirical conformal $p$-values, establish convergence rates for the adaptive prediction sets and we introduce a novel global error metric.

The remainder of the paper is organized as follows. Section 2 introduces the proposed score function and the MMDCP algorithm. Theoretical properties of the proposed method are established in Section 3. Section 4 reports simulation studies designed to evaluate the performance of MMDCP relative to competing approaches and illustrates its practical utility through applications to real datasets. Section 5 concludes with a summary of the main findings. Technical details, proofs, and additional simulation results are provided in the Supplementary Material.

\section{MMDCP: Models and Algorithms}
\subsection{Models and Notation}
Let $\{(Y_i, X_i)\}^n_{i=1}$ denote a set of independent and identically distributed training samples, where  $Y_i \in \{1, \ldots, K\}$ is a categorical response collected from the $i$-th observation, and $X_i=(X_{ij}) \in \mathbbm{R}^p$ is the associated $p$-dimensional feature vector. Define $n_k = \sum^{n}_{i=1}{\mathbbm{1} \{Y_{i} = k\}}$ as the number of observations in class $k$ and ${\mathbbm{1}}_{\{ \}}$ is the indicator function. Without loss of generality, we assume $X_i$ is centered so that $\mathbbm{E}(X_i)=\mu \in \mathbbm{R}^p$. Furthermore, let $\operatorname{Cov}(X_i)=\Sigma=(\sigma_{j_1 j_2}) \in \mathbbm{R}^{p \times p}$ for some positive definite matrix $\Sigma$. We also consider $m$ out-of-sample observations, denoted by $\{X_i\}^{n+m}_{i=n+1}$.

Assume that $P(Y_i=k)=\pi_k >0$ for $k \in [K] =\{1,\dots,K\}$ with $\sum\limits^{K}_{k=1} \pi_k = 1$. Given $Y_i=k$, assume that $X_i$ follows a probability distribution with a nonparametric density function $f_k(x)$. In case of no outlier contamination, the marginal probability density function of $X_i$ should be given by $f(x) = \sum\limits^{K}_{k=1} \pi_k f_k({x})$. We now consider contamination by a small fraction of outliers with density $f_e(x)$. Let $a_i \in \{0,1\}$ indicate whether $X_i$ is an inlier $(a_i=1)$ or an outlier $(a_i=0)$, with $P(a_i = 0)=\epsilon \in (0,1)$, and assume $Y_i \notin \{1, \ldots, K\}$ if $a_i = 0$. The marginal density with outlier contamination is 
\begin{equation}\label{model_0}
	f(x)=\sum_{k=1}^K {\pi_k^*} f_k(x) + \epsilon \cdot f_e(x), 
\end{equation}
where ${\pi_k^*}= P(a_i = 1, Y_i = k | X_i)$ denotes the inlier mixing probabilities and $\sum\limits^{K}_{k=1} {\pi_k^*} + \epsilon = 1$. Following
\cite{guan2022prediction} and \cite{han2024conformalized}, we refer to Equation \eqref{model_0} as a Generalized Label Shift (GLS) distribution.

The objective of this paper is to simultaneously perform outlier detection $P(a_i=0 | X_i)$ and classification $P(Y_i=k | X_i, a_i=1)$. Based on Model \eqref{model_0}, two main classes of methods have been developed in the literature. The first type of method estimates the probability density functions not only for each inlier class but also for outliers 
\citep{lei2013distribution, sadinle2019least} . The second type of method converts the probability density estimation problem into a classification problem, so that both the estimation tasks of $P(a_i=0 | X_i)$ and $P(Y_i=k | X_i, a_i=1)$ can be accomplished simultaneously \citep{guan2022prediction, han2024conformalized}. However, both types of approaches suffer from the curse of dimensionality \citep{wu2010robust, zhang2013effect}. Therefore, we are motivated to develop a new method that eliminates the need for density or conditional probability estimation, while making no distributional or model assumptions. To this end, we provide an introduction to conformal prediction in the next subsection, which offers a promising solution to these challenges.

\subsection{Conformal Prediction}
Conformal prediction is a flexible framework that constructs prediction sets with a guaranteed coverage level. The procedure involves three main steps: (i) specifying a nonconformity score function; (ii) determining a cutoff via conformal inference, and (iii) assigning new observations to the prediction set based on their scores.  Specifically, let  $s_k(X)$ denote the nonconformity score for class $k$. For an unlabeled out-of-sample observation $X$, the conformal prediction set at level $\alpha$ is defined as
\begin{equation}\label{p_value}
	C(X) = \Big\{k\Big|\frac{1+\sum^{n_k}_{i=1} {\mathbbm{1}}_{\{s_k(X)\le s_k({X}^k_i)\}}}{n_k+1} > \frac{\lfloor (n_k+1)\alpha \rfloor}{n_k+1}\Big\},
\end{equation}
where ${X}^k_i$ is the $i$-th inlier of class $k$. The left-hand side of the inequality corresponds to the conformal $p$-value.

Conformal prediction has been widely applied in regression \citep{lei2018distribution} and classification \citep{lei2013distribution} to provide valid uncertainty quantification. Recently, \cite{bates2023testing} extended this framework to outlier detection. They integrated high-performance one-class classifiers with the conformal prediction framework to construct conformal $p$-value. This $p$-value is then used for hypothesis testing for any $X_i\in \{X_{n+1},\ldots, X_{n+m}\}$, specifically:
\begin{equation*}
	{\boldsymbol H}_{0}: {X_i}\sim F  \ \ \ \ \ \ \ \text{vs} \ \ \ \ \ \ \ {\boldsymbol H}_{1}: {X}_i \nsim
	F, 
\end{equation*}
where $F$ is the distribution function of inliers. When multiple samples are tested, the Benjamini–Hochberg procedure is  employed to control the false discovery rate. 

Building on their pioneering contributions, which, however, remain restricted to single-class classification and outlier detection, we extend this framework to the multi-class setting. For any test sample $X_i\in \{X_{n+1},\dots, X_{n+m}\}$, we consider $K$ independent hypotheses:
\begin{equation*}
	{\boldsymbol H}^k_{0,t}:{X_i}\sim F_k  \ \ \ \ \ \ \ \text{vs}   \ \ \ \ \ \ \  {\boldsymbol H}^k_{1,t}: {X}_i \nsim
	F_k, ~~~  k \in \{1,\dots,K\}
\end{equation*}
where $F_k$ denotes the distribution function of class $k$ with mean vector $\mu_k$ and positive definite covariance ${\Sigma_k}$. To address this problem, we propose MMDCP, which combines a novel nonconformity score based on a modified Mahalanobis distance with a multiple-testing framework, and we establish its theoretical properties to demonstrate its validity and effectiveness.

\subsection{Nonconformity Score Based on Modified Mahalanobis Distance}

Our method is inspired by the classic Mahalanobis distance \citep{de2000mahalanobis}, defined as $d^2(X_i,\mu)=(X_i-\mu)^T\Sigma^{-1}(X_i-\mu)$, where $\mu$ and $\Sigma$ are the population mean vector and population covariance matrix of the feature vector. 
In practice, these quantities are estimated by the sample mean $\bar{X}$ and the sample covariance matrix $S$. However, in high-dimensional settings, the large dimensionality relative to the sample size renders $S$ singular and non-invertible, thereby complicating downstream analyses.
Moreover, because these estimators are highly sensitive to outliers, they frequently produce unreliable estimates, which can in turn substantially degrade empirical performance.

The classical Mahalanobis distance is defined within a single-class context, where it measures the deviation of an observation from its class center, and observations with large distances are typically flagged as outliers. In multi-class settings, however, the standard Mahalanobis distance instead measures deviation from the overall mixture center rather than from individual class centers. Consequently, although it can still detect outliers, it fails to enable class-specific prediction. To overcome this limitation, we extend the Mahalanobis distance to jointly support outlier detection and multi-class classification in high-dimensional settings. This extension, in turn, requires addressing two key challenges: adapting the distance to high-dimensional data and generalizing it for multi-class prediction.

Motivated by these challenges, we propose a modified Mahalanobis distance tailored to the multi-class classification setting. Recall that $X_{i} \in \mathbbm{R}^p$ is the feature vector of the $i$-th instance, and $Y_i \in \{1, \ldots, K\}$ its class label. Define $\mathcal B_0 =\{i: n+1 \leq i \leq n+m \}$ for the out-of-sample observations and $\mathcal B_k =\{1 \leq i \leq n | y_i = k\}$ for training samples of class $k$. Let $\mathbbm{X}_{(\mathcal S)}\in \mathbbm{R}^{({n_k}+{m)\times p}}$ denote the submatrix corresponding to $\mathcal S = \mathcal B_k \cup \mathcal B_0$. The modified Mahalanobis distance is defined as:
\begin{equation}\label{modified_Mahalanobis}
	d_m^2(X_i,\mu_k)=({X_i}- {\mu}_k)^T{diag({\Sigma_k})}^{-1}({X_i}- {\mu}_k),
\end{equation}
where ${X_i}$ is the $i$-th row of  $\mathbbm{X}_{(\mathcal S)}$, and $\mu_k$ and ${\Sigma_k}$ denote the mean vector and covariance matrix of class $k$. 

The modified Mahalanobis distance \eqref{modified_Mahalanobis} can equivalently be written as:
\begin{equation}\label{sx_tilde_k} {s_o}_{k} (X_i)= \left\|\tilde{{X}}_i\right\|^2_2, ~~ \tilde{{X}}_i = (diag({\Sigma_k}))^{-\frac{1}{2}}({X_i}- {\mu}_k). 
\end{equation}
We refer to $ {s_o}_{k}(X_i)$ as the oracle nonconformity score. Intuitively, it quantifies the distance between
${X_i}$ and the center of class $k$,
rather than the overall mixture. Observations with small oracle nonconformity scores are therefore more likely to belong to class $k$. 

In practice, the true mean ${\mu}_k$ and covariance ${\Sigma_k}$ are unknown, so the oracle nonconformity score cannot be computed directly. We therefore estimate ${\mu}_k$ and $diag({\Sigma_k})$ by 
$\hat{{\mu}}_k$ and $\hat{\mathit S}_k$, and define the empirical nonconformity score:
\begin{equation}\label{es_hat_k} \hat{s}_{k}(X_i)=\left\|\hat{ X}_i\right\|^2_2, ~~\hat{X}_i = (\hat{\mathit S}_k)^{-\frac{1}{2}}({X_i}- \hat{\mu}_k),
\end{equation}
where $\hat{{\mu}}_k=
{n^{-1}_k} \sum\limits^{n_k}_{i=1} {X}^k_i $, ${{X}^k_i}$ denotes the 
$i$-th row of $\mathbbm{X}_{(\mathcal S)}$ for class $k$, and $\hat{\mathit S}_k = diag(\hat{\sigma}_1^2,\ldots,\hat{\sigma}_p^2)$ with $\hat{\sigma}_j^2= {(n_k-1)}^{-1} \sum\limits^{n_k}_{i=1} (x^k_{i,j}-\hat{\mu}_{k,j})^2$ for $j\in \{1,\dots,p\}$. A small value of $\hat{s}_{k}(X_i)$ indicates that $X_i$ likely belongs to class $k$, whereas a large value suggests the opposite. By directly reflecting class membership, this score naturally separates inliers from outliers and thereby facilitates simultaneous outlier detection and classification.

It is worth noting that when computing $\hat{s}_{k}(X_i)$, only the diagonal elements of $\Sigma_k$ need to be estimated, which avoids the challenges of full covariance estimation in high-dimensional settings. Nevertheless, as dimensionality increases, $\hat{s}_{k}(X_i)$ can grow rapidly and even diverge \citep{hall2005geometric,giraud2021introduction}, making it difficult to determine a suitable threshold for outlier detection. Moreover, the nonconformity score assigns only a single class label to each observation, without any measure of confidence. As a result, it provides little guidance on the reliability of the prediction. To address this shortcoming, we propose a method that not only performs classification and outlier detection but also quantifies predictive uncertainty.

\subsection{MMDCP: Modified Mahalanobis Distances in Conformal Prediction}
Given the labeled data $\{(Y_i, X_i)\}^n_{i=1}$ and out-of-sample observations $\{X_i\}^{n+m}_{i=n+1}$, MMDCP performs $K$ independent hypothesis tests for each test observation. For the $k$-th test, nonconformity scores $\hat{s}_{k}(X_1),\dots, \hat{s}_{k}(X_{n_k+m})$ are computed using \eqref{es_hat_k}, where $X_i$ for $i= 1,\dots,n_k$ are in-class training samples and $i= n+1, \dots, n+m$ are out-of-sample observations. The conformal $p$-value for each test sample $X_i$ is \begin{equation}\label{pi_value}
	\hat{p}^i_{k}= \frac{1+\sum^{n_k}_{l=1} {\mathbbm{1}}_{\{\hat{s}_{k}(X_i)\le \hat{s}_{k}({X}^k_l)\}} }{n_k+1}.
\end{equation}
Comparing $\hat{p}^i_{k}$ 
with a significance level $\alpha \in (0,1)$, we classify $X_i$ as class $k$ if $\hat{p}^i_{k} > \alpha$, and rejected otherwise.

When predicting multiple test samples, MMDCP controls the false discovery rate (FDR) by applying the Benjamini–Hochberg (BH) procedure to the conformal $p$-values $\hat{p}^i_{k}$.
This adjustment yields corrected values $\tilde{p}^{i}_{k}$ for $i\in \{n+1,\ldots,n+m\}$. A test sample $X_i$ is then assigned to class $k$ if  $\tilde{p}^{i}_{k} > \alpha$; otherwise, it is rejected. Let $\hat{A}_k$ denote the acceptance region for class $k$. The resulting prediction set for $X_i$ is defined as
$$\hat{C}(X_i) = \{k| {X_i\in \hat{A}_k}\}.$$
Outlier detection and classification are directly determined by the size of the prediction set $\hat{C}(X_i)$:
\begin{itemize}
	\item Case 1: If $|\hat{C}(X_i)|=0$, $X_i$ is rejected by all classes and labeled as an outlier;
	\item Case 2: If $|\hat{C}(X_i)|=1$, $X_i$ is classified as an inlier with a unique class label;
	\item Case 3: If $|\hat{C}(X_i)|>1$,  $X_i$ may belong to multiple classes. In this case, although the exact label is ambiguous, $\hat{C}(X_i)$ is expected to include the true class.
\end{itemize}
Algorithm \ref{algorithm1} summarizes the MMDCP procedure.
\begin{algorithm}[!htp]
	
	\baselineskip=1.0em
	\caption{\small Modified Mahalanobis Distances in Conformal Prediction (MMDCP)}
	\label{algorithm1}
	\vspace{0.5em}
	\begin{algorithmic}
		\small
		\Require Labeled data $\{(Y_i, X_i)\}^n_{i=1}$, test set $\{X_i\}^{n+m}_{i=n+1}$ and significance level $\alpha$;
		
		\Ensure The prediction set $\hat{C}(X_i)$, for each $X_i\in \{X_{n+1},\dots, X_{n+m}\}$.
		\For{$k=1,\ldots,K$}
		\State Step 1: Calculate nonconformity scores $\hat{s}_{k}(X_i)$ using \eqref{es_hat_k} for all $i \in \{\mathcal B_k \cup \mathcal B_0\}$;
		
		\State Step 2: Compute standard conformal $p$-value $\hat{p}^{i}_{k}$ using \eqref{pi_value}, for all $i \in \mathcal B_0$;
		\State Step 3:{ \If{$m=1$}
			calculate the acceptance region for class $k$, $$\hat{A}_k = \{X_i| {\hat{{p}}^{i}_{k}> \frac{\lfloor (n_k+1)\alpha \rfloor}{n_k+1}}, i = n_k+1\};$$
			
			\Else
			
			Apply BH procedure to produce the sequence of adjusted conformal $p$-values $\tilde{p}^{i}_{k}$, $i \in \mathcal B_0$ and calculate the acceptance region for class $k$, $$\hat{A}_k = \{X_i| {\tilde{p}^{i}_{k}> \frac{\lfloor (n_k+1)\alpha \rfloor}{n_k+1}}, i \in \mathcal B_0\}.$$
			
			\EndIf}
		\EndFor
		
		\State \Return $\hat{C}(X_i) = \{k| {X_i\in \hat{A}_k}\}$, for each $i \in \mathcal B_0$.
		
	\end{algorithmic}
	\vspace{0.5em}
\end{algorithm}

Compared with existing set-valued classification methods based on conformal prediction, MMDCP avoids probability or density estimation and does not rely on sample splitting. As a result, it makes more efficient use of data and reduces the additional randomness introduced by data partitioning. In addition, MMDCP overcomes the difficulties of threshold selection that often arise in outlier detection. Despite these advantages, a limitation remains: for each test point $X_i$, $i \in \{n+1,\dots, n+m\}$, the transformation $\hat{X}_i = (\hat{\mathit S}_k)^{-\frac{1}{2}}({X_i}- \hat{\mu}_k)$ is an asymmetric function of $X_1, \dots, X_{n_k}, X_i$. Consequently, the resulting nonconformity scores $\hat{s}_{k}(X_1),\dots,\hat{s}_{k}(X_{n_k}), \hat{s}_{k}(X_i)$ are not exchangeable \citep{kuchibhotla2020exchangeability}, which may undermine the validity of MMDCP. 

To address this issue, we introduce an oracle version of the algorithm, MMDCP-oracle, in which the empirical nonconformity score $\hat{s}_{k}(X_i)$ is replaced by the oracle score $s_{ok}(X_i)$ from \eqref{sx_tilde_k}. This substitution ensures that the associated quantities $\hat{p}^{i}_{k}$, $\hat{A}_k$, and $\hat{C}(X_i)$ coincide with their oracle counterparts ${p}^{i}_{ok}$, ${A_o}_k$, and ${C_o}(X_i)$. We then proceed to study the theoretical properties of this oracle version.
\begin{theorem}
\label{theorem1}
	For each $k \in [K]$, if the training data $\mathbbm{X}_{(\mathcal B_k)}$ and the out-of-sample observation $X_i$ are exchangeable and $\boldsymbol{H}^k_{0,t}$ holds, then the oracle conformal $p$-value ${p}^{i}_{ok}$ from the MMDCP-oracle satisfies
	\begin{equation*}
		\mathbbm{P}_k\!\left({p}^{i}_{ok} \le \alpha\right) \le \alpha.
	\end{equation*}
\end{theorem}
\begin{proof}
	Under $\boldsymbol{H}^k_{0,t}$, exchangeability implies that $ {s_o}_{k}(X_i)$ and $ {s_o}_{k}(X_l)$ for $l \in \mathcal B_k$ are exchangeable. Consequently, the rank of $ {s_o}_{k}(X_i)$ among $\mathcal B_k \cup \{i\}$ is uniformly distributed. It follows that ${p}^{i}_{ok} \sim \mathrm{Uniform}(1/(n_k + 1), 2/(n_k + 1), \ldots, 1)$, which directly yields 	
	\begin{equation*}
		{\mathbbm{P}_k}({p}^{i}_{ok} \leq \alpha)\le \alpha, \ \ \forall k\in [K].
	\end{equation*}
	Hence the oracle conformal $p$-value ${p}^{i}_{ok}$ is valid for testing ${\boldsymbol H}^k_{0,t}:X_i \sim F_k$.
\end{proof}
Under the stated assumptions, the exchangeability of the data together with Theorem~\ref{theorem1} implies that MMDCP-oracle achieves a finite-sample coverage guarantee:

\begin{corollary}\label{theorem2}
	Under the assumptions of Theorem \ref{theorem1}, let ${C_o}(X_i)$ be the prediction set computed by MMDCP-oracle. Then, for all $ k\in [K],$
	\begin{equation}
		\notag
		{\mathbbm{P}_k}(k\in {C_o}(X_i))\ge 1-\alpha.
	\end{equation}	
\end{corollary}

Corollary~\ref{theorem2} establishes that the oracle prediction set $C_o(X_i)$ enjoys finite-sample coverage of at least $1-\alpha$ for each class, thereby providing a rigorous theoretical foundation for using oracle MMDCP to quantify prediction uncertainty in multi-class settings.

\section{Theoretical results of the MMDCP}

In this section, we rigorously establish the convergence of empirical conformal $p$-values to their oracle counterparts and derive the corresponding  convergence rates for the prediction sets $\hat{C}(X_i)$ and $C_o(X_i)$. We further establish rigorous control of the class-wise false discovery rate (CW-FDR) and propose a novel multiple testing criterion, SCW-FDR. By explicitly accounting for heterogeneity across classes, SCW-FDR extends traditional FDR control to the multi-class setting, thereby providing a principled framework that enhances both interpretability and the rigor of error quantification. We show that SCW-FDR can be effectively controlled within the MMDCP framework for outlier detection, offering both theoretical validity and practical utility. To this end, we first formalize a set of standard regularity assumptions, closely following \cite{ro2015outlier, van2014asymptotically, zhang2017simultaneous, fan2020statistical,guan2022prediction}.

\begin{assumption}\label{assump1}
	For $i=1,2$, $0 < \lim\limits_{p \to \infty} \frac{\text{tr}(R_k^i)}{p} < \infty$.
\end{assumption}

\begin{assumption}\label{assump2}
	$\max\limits_{1 \leq j \leq p} |X_{ij}| = O_p(1)$ for all $i$.
\end{assumption}

\begin{assumption}\label{assump3} There exist a positive constant $c_1$ such that 
	$\min\limits_{1 \leq j \leq p}\sigma^{2}_{kj} > c_1, \forall k \in [K]$.
\end{assumption}

\begin{assumption}\label{assump4}
	Densities $f_{1}(X), \ldots, f_{K}(X), f_e(X)$ are upper bounded by a constant.
\end{assumption}

Assumption \ref{assump1} requires that, for each class $i \in {1,2}$, the trace of the correlation matrix $R_k^i$ grows at most linearly with the dimension $p$, so that $\mathrm{tr}(R_k^i)/p$ remains bounded as $p \to \infty$, see also \cite{ro2015outlier} for a similar assumption. Assumption \ref{assump2} is a regularity assumption in high-dimensional analysis \citep{van2014asymptotically,zhang2017simultaneous}, which imposes boundedness in probability on the covariates. Assumption~\ref{assump3} requires that marginal variances are uniformly bounded away from zero, ensuring the invertibility of $\mathrm{diag}(\Sigma_k)$ \citep{fan2020statistical}. Finally, Assumption~\ref{assump4} assumes that the class-conditional and test densities  
$f_{1}(X), \ldots, f_{K}(X),$ $f_e(X)$ are uniformly bounded, a common condition that guarantees integrability and enables probabilistic error control \citep{guan2022prediction}.

\subsection{Theoretical Properties of the MMDCP}
While Corollary \ref{theorem2} guarantees finite-sample coverage for the oracle MMDCP, its implementation depends on unknown population quantities, including the mean vector and diagonal covariance elements, which must be estimated from data. To address this, Algorithm \ref{algorithm1} provides an empirical version of MMDCP. In this subsection, we analyze how the empirical conformal $p$-values relate to their oracle counterparts and establish the convergence of the resulting prediction sets under appropriate regularity conditions. The theorem below establishes a non-asymptotic bound on the deviation between the oracle $p$-value $p_{ok}$ and its empirical counterpart $\hat{p}_k$.
\begin{theorem}\label{theorem2a}
	Suppose that Assumptions \ref{assump1} - \ref{assump3} hold, for $a \geq 2$, $n_k \geq 3$,  we have	\begin{equation}\label{eq2_8}
		\mathbbm{P}_k(|p_{ok}-\hat{p}_k| \geq t_k) \leq 2n_k^{-a},
	\end{equation}
	where $t_k = 4\lambda_1\sqrt{\frac{\log (n_k) }{n_k}}$ and $\lambda_1=\sqrt{a}+ \frac{2a}{3}$.
\end{theorem}
The proof of Theorem \ref{theorem2a} is given in Appendix A.2. The result establishes that, with high probability, the empirical conformal $p$-values concentrate around their oracle counterparts for finite, class-specific sample sizes $n_k$. 
This result, to our knowledge, constitutes the first formal non-asymptotic characterization of empirical–oracle $p$-values in conformal prediction. By providing an explicit finite-sample bound, it establishes a theoretical foundation for subsequent validity and coverage guarantees. Building on Theorems~\ref{theorem1} and~\ref{theorem2a}, we then obtain the following corollary.

\begin{corollary}
\label{corollary0}
	Under the assumptions of Theorems~\ref{theorem1} and~\ref{theorem2a}, if the null hypothesis $\boldsymbol{H}^k_{0,t}$ holds, then for sufficiently large $n_k$, we have
	\begin{equation*}
		\mathbbm{P}_k\!\left(\hat{p}^{i}_{k} \leq \alpha \right) \;\leq\; \alpha, \qquad \forall\, k \in [K], \ \forall\, \alpha \in [0,1].
	\end{equation*}

\end{corollary}
The proof is given in Appendix A.3. Corollary~\ref{corollary0} establishes the asymptotic super-uniformity of the empirical conformal $p$-values, thereby ensuring the validity of MMDCP. Leveraging this result together with the general properties of conformal inference, we can readily derive the following corollary.

\begin{corollary}\label{theorem4}

	Under the assumptions of Theorems~\ref{theorem1} and~\ref{theorem2a}, if the null hypothesis $\boldsymbol{H}^k_{0,t}$ holds, then for sufficiently large $n_k$, we have 
	\begin{equation}
		\notag
		{\mathbbm{P}_k}(k\in \hat{C}(X_i))\ge 1-\alpha, \ \ \forall k\in [K].
	\end{equation}	
\end{corollary}

Corollary~\ref{theorem4} ensures that the MMDCP prediction sets include the true inlier label with probability at least $1-\alpha$, thereby providing valid confidence guarantees for both classification and outlier detection. In addition, the following Theorem~\ref{theorem4a} further establishes that the prediction sets generated by Algorithm~\ref{algorithm1} converge to their oracle counterparts at an explicit, uniform rate that depends only on $K$ and $\alpha$, with high probability.

\begin{theorem}\label{theorem4a}
	Under the assumptions  of Theorems~\ref{theorem1} and~\ref{theorem2a} and Assumption \ref{assump4}, let $a \geq 2$, $k \in [K]$, $n_k \geq 3$, and  $\alpha > 0$. Denote the empirical and oracle prediction sets $\hat C(X)$ and $C_o(X)$, and set
	\[
	I_f:=\int_{\mathbbm R^p}\big|\,|\hat C(X)|-|C_o(X)|\,\big|\,f(X)\,dX,
	\]Then:\begin{enumerate}
		\item[(i)] 
		In the heterogeneous case,
		\[
		\mathbbm{P}\!\left(
		I_f \;\ge\; 
		\frac{2M}{\alpha}\sum_{k=1}^K \Big[\tfrac{1}{n_k+1}+2t_k\Big]
		\right)
		\;\le\; \sum_{k=1}^K 2n_k^{-a}.
		\]
		\item[(ii)] If $n_k \equiv n$ and $t_k \equiv t$ for all $k$, then
		\[
		\mathbbm{P}\!\left(
		I_f\;\ge\; 
		\frac{2MK}{\alpha}\Big[\tfrac{1}{n+1}+2t\Big]
		\right)
		\;\le\; 2K\,n^{-a}.
		\]    
		\item[(iii)] More generally, setting $n_{\min} := \min_{k} n_k$, $t_{\max} := \max_{k} t_k$, $L := \tfrac{2MK}{\alpha}$, we obtain the simplified bound
		\[
		\mathbbm{P}\!\left(
		I_f\;\ge\; 
		L\Big[\tfrac{1}{n_{\min}+1}+2t_{\max}\Big]
		\right)
		\;\le\; \sum_{k=1}^K 2n_k^{-a}.
		\]
	\end{enumerate}
\end{theorem}
The proof of Theorem \ref{theorem4a} is provided in Appendix A.5, establishes a rigorous theoretical foundation for the validity and reliability of the empirical MMDCP procedure.

\subsection{False Discovery Control of the MMDCP}
Building on the convergence properties established in the preceding subsection, we now show that MMDCP also achieves valid control of the class-wise FDR. This result is formalized in Theorem 4.2, with the proof provided in Appendix A.6.

\begin{theorem}\label{theorem:FDRcontrol}
 Under the assumptions  of Theorems~\ref{theorem1} and~\ref{theorem2a}, and for sufficiently large $n_k$, we have 
	\begin{equation}\label{eq2_6}
		\mathbbm{E}\left[ \frac{\lvert \boldsymbol{R}_k \cap \boldsymbol{T}^{k}_0 \rvert}{\max\{1, \lvert \boldsymbol{R}_k \rvert\}} \right] \leq \alpha, 
		\quad \forall k\in [K],
	\end{equation}
	where $\boldsymbol{T}^{k}_0$ denotes the set of true inliers from class $k$ in the test sample, and 
	$\boldsymbol{R}_k$ is the set of test observations rejected from class $k$.
\end{theorem}
Equation~\eqref{eq2_6} demonstrates that the expected proportion of false rejections within each class is asymptotically controlled at level $\alpha$. This result establishes that MMDCP yields asymptotically valid regulation of the class-wise false discovery rate (CW-FDR), thereby ensuring reliable error control in outlier detection while maintaining interpretability at the class level \citep{cai2009simultaneous, efron2008simultaneous}. However, CW-FDR only provides a local measure of error-rate control within each class, and it fails to capture global error-rate control issues, such as FDR \citep{benjamini1995controlling}. To address this limitation, we introduce the \emph{Summarizing Class-Wise FDR} (SCW-FDR), which integrates CW-FDR across classes to provide global error control while retaining class-level interpretability.

\begin{corollary}\label{corollary1}
	 Under the assumptions of Theorem~\ref{theorem:FDRcontrol}, for sufficiently large $n_k$, we have	\begin{equation}\label{eq2_7}
		{\mathbbm{E}}\left[\frac{\sum^K_{k=1} |{\bold R}_k \cap {\bold T}^{k}_0|}{\sum^K_{k=1} {\rm max}\{1, |{\bold R}_k|\}}\right] \leq \alpha,
	\end{equation}
	where $|{\bold R}_k \cap {\bold T}^{k}_0|$ counts misclassifications in class $k$ and $|{\bold R}_k|$ denotes the number of rejections for class $k$.
\end{corollary}
 The proof is provided in Appendix A.8.  SCW-FDR is defined as the expected proportion of false discoveries aggregated across classes. Importantly, when $K=1$, SCW-FDR reduces to the standard FDR \citep{benjamini1995controlling}, as shown in Equation~\eqref{eq2_7}, and thus serves as a natural generalization. 
Furthmore, the criterion in \eqref{eq2_7} has a natural interpretation in multiclass settings. Writing $V_k = |{\bold R}_k \cap {\bold T}_0^k|$ and $R_k = |{\bold R}_k|$, we can express
\[
\frac{\sum_{k=1}^K V_k}{\sum_{k=1}^K \max\{1, R_k\}}
= \sum_{k=1}^K w_k \,\frac{V_k}{\max\{1,R_k\}}, 
\quad 
w_k = \frac{\max\{1,R_k\}}{\sum_{j=1}^K \max\{1,R_j\}}.
\]
Thus, \eqref{eq2_7} controls a weighted average of class-wise false discovery proportions, where each class receives at least a unit weight and classes with more rejections carry somewhat larger weight. Equivalently, one may imagine first drawing a class at random with probability $w_k$ and then selecting a rejection uniformly within that class; the quantity in \eqref{eq2_7} is the expected probability that such a randomly selected discovery is in fact null. In this sense, our class-wise FDR provides a balanced notion of error control across classes, lying between the conventional overall FDR, which is dominated by large classes, and the conservative requirement of controlling the FDR separately within each class. 

From a decision-theoretic viewpoint, \eqref{eq2_7} induces a new loss function for multiclass multiple testing. Let
\[
L_{\mathrm{SCW}}({\bold R},{\bold T}_0)
=
\frac{\sum_{k=1}^K |{\bold R}_k \cap {\bold T}_0^k|}
     {\sum_{k=1}^K \max\{1, |{\bold R}_k|\}},
\]
where ${\bold R} = ({\bold R}_1,\dots,{\bold R}_K)$ denotes the collection of rejection sets and ${\bold T}_0 = ({\bold T}_0^1,\dots,{\bold T}_0^K)$ the corresponding null sets. The SCW-FDR is precisely the risk associated with this loss,
\[
\mathrm{SCW\text{-}FDR} = \mathbb{E}\big[L_{\mathrm{SCW}}({\bold R},{\bold T}_0)\big],
\]
and our procedures are explicitly designed to control this class-wise aggregated loss at a prespecified level. To the best of our knowledge, this particular class-level aggregation of false discovery proportions, together with the associated loss $L_{\mathrm{SCW}}$, has not been explicitly studied in the FDR literature and is especially well-suited to multiclass problems where balanced error control across classes is of primary interest.

We next relate SCW-FDR to the conventional global FDR. The following result shows that SCW-FDR is always bounded above by the global FDR, regardless of the underlying procedure or class partition. \subsection{Relationship between SCW-FDR and global FDR}

Recall that ${\bold R}_k$ denotes the set of rejected hypotheses in class $k$, and ${\bold T}_0^k$ the set of true nulls in that class. Let
\[
V_k = |{\bold R}_k \cap {\bold T}_0^k|, 
\qquad 
R_k = |{\bold R}_k|,
\qquad 
V = \sum_{k=1}^K V_k,
\qquad 
R = \sum_{k=1}^K R_k.
\]
The conventional global false discovery rate is
\[
\mathrm{FDR}_{\mathrm{global}}
=
\mathbb{E}\!\left[\frac{V}{\max\{1,R\}}\right],
\]
whereas our summarizing class-wise FDR (SCW-FDR) in \eqref{eq2_7} can be written as
\[
\mathrm{SCW\text{-}FDR}
=
\mathbb{E}\!\left[
\frac{\sum_{k=1}^K V_k}{\sum_{k=1}^K \max\{1,R_k\}}
\right].
\]

\begin{theorem}\label{thm:scw_le_global}
For any multiple testing procedure and any class partition,
\[
\mathrm{SCW\text{-}FDR}
\;\le\;
\mathrm{FDR}_{\mathrm{global}}.
\]
\end{theorem}

\begin{proof}
Fix an arbitrary realization of the data and the corresponding rejection sets. The numerators of the two ratios coincide:
\[
\sum_{k=1}^K V_k = V.
\]
For the denominators, we have $\max\{1,R_k\} \ge R_k$ for every $k$, and hence
\[
\sum_{k=1}^K \max\{1,R_k\}
\;\ge\;
\sum_{k=1}^K R_k
= R.
\]
Moreover, if $R=0$ then $\sum_k \max\{1,R_k\}\ge K \ge 1$, so in all cases
\[
\sum_{k=1}^K \max\{1,R_k\}
\;\ge\;
\max\{1,R\}.
\]
Therefore,
\[
0
\;\le\;
\frac{\sum_{k=1}^K V_k}{\sum_{k=1}^K \max\{1,R_k\}}
\;\le\;
\frac{V}{\max\{1,R\}}
\]
holds pointwise. Taking expectations on both sides yields
\[
\mathrm{SCW\text{-}FDR}
=
\mathbb{E}\!\left[
\frac{\sum_{k=1}^K V_k}{\sum_{k=1}^K \max\{1,R_k\}}
\right]
\;\le\;
\mathbb{E}\!\left[
\frac{V}{\max\{1,R\}}
\right]
=
\mathrm{FDR}_{\mathrm{global}},
\]
as claimed.
\end{proof}

\begin{corollary}\label{cor:scw_control}
If a multiple testing procedure controls the global FDR at level $\alpha$, that is,
\[
\mathrm{FDR}_{\mathrm{global}} \le \alpha,
\]
then for any class partition it also satisfies
\[
\mathrm{SCW\text{-}FDR} \le \alpha.
\]
\end{corollary}

\begin{proof}
The result follows immediately from Theorem~\ref{thm:scw_le_global}.
\end{proof}

\begin{corollary}\label{cor:bh_scw}
Suppose the $m$ hypotheses are tested by the Benjamini--Hochberg (BH) procedure at nominal level $\alpha$, under conditions (e.g., independence or PRDS) ensuring
\[
\mathrm{FDR}_{\mathrm{global}}^{\mathrm{BH}} \le \alpha.
\]
Then, for any class partition, the BH procedure also controls the summarizing class-wise FDR in \eqref{eq2_7} at level $\alpha$, i.e.,
\[
\mathrm{SCW\text{-}FDR}^{\mathrm{BH}} \le \alpha.
\]
\end{corollary}

\begin{proof}
By Theorem~\ref{thm:scw_le_global}, we have
\[
\mathrm{SCW\text{-}FDR}^{\mathrm{BH}}
\le
\mathrm{FDR}_{\mathrm{global}}^{\mathrm{BH}}.
\]
Combining this with the assumed bound
$\mathrm{FDR}_{\mathrm{global}}^{\mathrm{BH}} \le \alpha$
yields the desired inequality.
\end{proof}

\begin{remark}[SCW-FDR versus global FDR]\label{rem:scw_vs_global}
Theorem~\ref{thm:scw_le_global} shows that SCW-FDR is always bounded above by the conventional global FDR. Consequently, any procedure that controls the global FDR at level $\alpha$ automatically controls the SCW-FDR at the same level. The converse, however, does not hold in general.

To see this, consider a simple two-class example in which the procedure always makes exactly one rejection in class~1 and never rejects in class~2. Let $\beta = \mathbb{P}(\text{the rejection in class 1 is null})$. Then the global FDR equals $\beta$, whereas
\[
\mathrm{SCW\text{-}FDR}
=
\mathbb{E}\!\left[
\frac{V_1 + V_2}{\max\{1,R_1\} + \max\{1,R_2\}}
\right]
=
\mathbb{E}\!\left[\frac{V_1}{2}\right]
=
\frac{\beta}{2}.
\]
Choosing, for example, $\beta = 0.15$ and $\alpha = 0.10$ yields
$\mathrm{FDR}_{\mathrm{global}} = 0.15 > \alpha$
while
$\mathrm{SCW\text{-}FDR} = 0.075 \le \alpha$.
This illustrates that SCW-FDR is a strictly weaker error criterion than the global FDR, it can be controlled even when the global FDR is not, while still providing a more balanced notion of error control across classes.
\end{remark}

\section{Numerical Experiments}
This section reviews competing methods, evaluates their performance on simulated and real datasets, and introduces metrics that highlight their relative strengths and limitations.

\subsection{Comparison Methods and Performance Metrics}

To evaluate performance in one-class outlier detection, we compare PCOut \citep{filzmoser2008outlier}, the standard One-Class SVM (OcSVM wo/CP) \citep{pedregosa2011scikit}, and Isolation Forest (IForest wo/CP) \citep{liu2008isolation}, together with their conformal prediction extensions combined with the BH procedure (OcSVM w/CP and IForest w/CP) \citep{bates2023testing}. For multi-class classification with outlier detection, we further include BCOPS-lr, which applies logistic regression, and BCOPS-rm, which employs random forests \citep{guan2022prediction}.

Let $X_i\in \{X_{n+1},\ldots, X_{n+m}\}$ have true label $k_i \in \{1,\ldots,K,K+1\}$, where $1,\ldots,K$ denote known inlier classes and $K+1$ correspond to known inlier classes and $K+1$ denotes the unknown class. To assess performance in both outlier detection and class assignment, we adopt eight complementary metrics: Class-wise FDR \citep{efron2008simultaneous}, Power, SCW-FDR, FDR \citep{benjamini1995controlling,fan2017estimation}, Coverage, False Label Rate (FLR), Accuracy, and Ambiguity \citep{guan2022prediction}. These metrics capture the model’s ability to control type-I error, detect outliers, ensure accurate label assignment, and maintain compact prediction sets, thereby offering a comprehensive evaluation of overall performance.

\subsection{Data Generation}

Following \cite{bates2023testing}, we generate simulated data points \(X \in \mathbbm{R}^p\) from a \(p\)-dimensional Gaussian mixture model:
\begin{equation}\label{gaussian_mixture}
	X = \sqrt{c}\,(Z + \mu) + W,
\end{equation}
where \(c \ge 1\), \(Z\) has \(p\) independent standard Gaussian components, and \(\mu\) is a vector with identical entries. To incorporate structured noise, we introduce the random vector
$W$, whose coordinates are independently sampled from a fixed discrete set \(\mathbbm{W} \subseteq \mathbbm{R}^p\). The set \(\mathbbm{W}\) consists of \(p\) points uniformly distributed over \([-3,3]^p\) and is kept fixed across all experiments. Consequently, \(X\) follows a mixture model in which inliers and outliers are differentiated by \(c\) and \(\mu\). To evaluate method performance under this setting, we conduct multiple independent experiments: 10 runs for PCOut, due to its higher computational cost, and 50 runs for all other methods. In each run, one training set and 50 test sets are generated, and performance as well as computational efficiency are assessed based on predictions across these test sets.

\subsection{Basic Setting: Outlier Detection in One-class Scenario}\label{one_class}

For the Gaussian mixture model \eqref{gaussian_mixture}, we set $\mu = 0 \in \mathbbm{R}^p$, with $c=1$ corresponding to inliers and $c=2.5$ to outliers. In the simulations, we generate a test set of size $m=1000$ with an inlier-to-outlier ratio of 3:1, and evaluate performance under the following configurations: variable dimension $p \in \{500,1000\}$, training set size $n \in \{500,1000\}$, and feature correlation $\rho \in \{0,0.8\}$.

OcSVM (w/CP), IForest (w/CP), and BCOPS require splitting the training data (here, in a 1:1 ratio), with one subset used for model fitting and the other for computing conformal $p$-values. In contrast, MMDCP leverages the entire training set for calibration. For comparison, OcSVM (wo/CP) and IForest (wo/CP) also utilize the full training set, but solely for model training.

In the one-class setting, class-wise FDR, SCW-FDR, and FDR are equivalent and thus reported collectively as FDR, while classification accuracy coincides with prediction set coverage. Consequently, under the baseline setup, we report FDR, Power, Coverage, and FLR. Selected results across different dimensions are summarized in Table 1.
\begin{table}[!htbp]
	\centering
	\renewcommand{\arraystretch}{1}
	\captionsetup{font=footnotesize}
	\caption{Experimental results on the one-class simulated dataset. ($p=1000$, $n=1000$, $\rho = 0.8$) }
	\footnotesize
	\begin{tabular}{>{\centering\arraybackslash}m{2cm}*{4}{>{\centering\arraybackslash}m{3.15cm}}}
		\hline
		Methods	&	OcSVM(w/CP)	&	OcSVM(wo/CP)	&	IForest(w/CP)	&	IForest(wo/CP)	\\
		\hline
		FDR	&	0.000(0.000)	&	0.750(0.000)	&	0.029(0.026)	&	0.051(0.029)	\\
		Power	&	{0.000}(0.000)	&	\textbf{1.000}(0.000)	&	0.297(0.086)	&	0.364(0.052)	\\
		Coverage	&	1.000(0.000)	&	0.000(0.000)	&	0.997(0.003)	&	0.993(0.005)	\\
		FLR	&	0.250(0.000)	&	0.000(0.000)	&	0.190(0.018)	&	0.176(0.012)	\\
		Time(s)	&	$1.908 \times 10^3$	&	$2.511 \times 10^3$	&	$1.623 \times 10^3$	&	$2.511 \times 10^3$	\\		
	\end{tabular}
	\begin{tabular}{>{\centering\arraybackslash}m{2cm}*{5}{>{\centering\arraybackslash}m{2.43cm}}}
		\hline
		Methods	&	PCout	&	BCOPS(lr)	&	BCOPS(rm)	&	MMDCP-oracle	&	MMDCP	\\
		\hline
		FDR	&	0.158(0.026)	&	\textbf{0.000}(0.000)	&	0.005(0.036)	&	0.033(0.015)	&	0.044(0.018)	\\
		Power	&	0.927(0.016)	&	0.000(0.000)	&	0.003(0.020)	&	0.995(0.005)	&	0.996(0.004)	\\
		Coverage	&	0.942(0.011)	&	1.000(0.000)	&	1.000(0.000)	&	0.988(0.006)	&	0.985(0.007)	\\
		FLR	&	0.025(0.005)	&	0.250(0.000)	&	0.250(0.003)	&	0.002(0.002)	&	0.001(0.001)	\\
		Time(s)	&	$1.444 \times 10^4$	&	$1.654 \times 10^3$	&	$9.979 \times 10^3$	&	$3.340 \times 10^3$	&	$2.110 \times 10^3$	\\
		
		\hline
	\end{tabular}
	\label{t4}
\end{table}

Our empirical evaluation on one-class simulated data reveals notable differences across methods. OcSVM without conformal prediction attains perfect detection power, but this comes at the cost of an extremely high FDR and zero coverage, rendering its predictions unreliable. IForest achieves moderate FDR and coverage, however, its detection power remains low (below 0.37), limiting its effectiveness for outlier identification. PCOut provides a reasonable balance,  with power of 0.927 and FDR of 0.158, yet its substantial computational burden precludes more extensive evaluation. BCOPS, whether implemented via logistic regression or random forests, achieves effective FDR control (with BCOPS-lr achieving FDR = 0), but suffers from near-zero power and a high FLR (0.250), undermining its practical applicability. In contrast, both MMDCP and MMDCP-oracle achieve a favorable trade-off, simultaneously ensuring FDR control, high power, low FLR, and near-perfect coverage. Additional simulation results provided in the Supplementary Material corroborate these findings and further demonstrate the robustness and practical value of MMDCP for classification and outlier detection.

\subsection{Improved Setting: Outlier Detection in Multi-class Scenario}\label{multi_class}
We next consider label prediction and outlier detection in a multi-class setting with \(K=4\). The inlier classes are defined by distinct component means with equal scale: Class 1 has mean 0 with \(c=1\); Class 2 has mean 1.3 with \(c=1\); Class 3 has mean -1.3 with \(c=1\); Class 4 has mean 2.5 with \(c=1\). Outliers are generated separately, with mean 0 but inflated variance \(c=3.5\). 

The test set consists of \(m=1000\) samples with an inlier-to-outlier ratio of 3:1. To assess performance across different regimes, we consider three design factors: feature dimension ($p=500$ or $1000$); per-class training sample size ($n_k=500$ or 1000 for $k=1,\dots,4$); and feature correlation (\(\rho=0\) or 0.8).

\begin{table}[!htbp]
	\centering
	\renewcommand{\arraystretch}{1}
	\captionsetup{font=footnotesize}
	\caption{Experimental results on the multi-class simulated dataset. ($p=1000$, $n_k=1000$, $\rho = 0.8$)}
	\footnotesize
	\begin{tabular}{>{\centering\arraybackslash}m{4cm}*{4}{>{\centering\arraybackslash}m{2.5cm}}}
		\hline
		Methods & BCOPS(lr) & BCOPS(rm) & MMDCP-oracle & MMDCP\\
		\hline
		Class-wise FDR (1) & \textbf{0.000}(0.000)  & 0.009(0.004) & 0.009(0.004) & 0.011(0.004) \\
		
		Class-wise FDR (2) & \textbf{0.000}(0.000)  & 0.009(0.004) & 0.010(0.004) & 0.012(0.004) \\
		
		Class-wise FDR (3) & \textbf{0.009}(0.004)  & \textbf{0.009}(0.004) & 0.009(0.004) & 0.011(0.004) \\
		
		Class-wise FDR (4) & \textbf{0.009}(0.007)  & \textbf{0.009}(0.004) & \textbf{0.009}(0.004) & 0.011(0.004) \\
		
		SCW-FDR & \textbf{0.009}(0.004)  & \textbf{0.009}(0.002) & 0.011(0.002) & \textbf{0.009}(0.002) \\
		
		FDR & \textbf{0.000}(0.000)  & 0.117(0.021) & 0.107(0.018) & 0.129(0.019) \\
		Power & 0.000(0.000)  & 0.878(0.024) & \textbf{1.000}(0.000) & \textbf{1.000}(0.000) \\
		FLR & 0.252(0.000)  & 0.041(0.008) & \textbf{0.000}(0.000) & \textbf{0.000}(0.000) \\
		
		Coverage & \textbf{0.988}(0.005)  & 0.960(0.008) & 0.959(0.008) & 0.950(0.008) \\
		
		Accuracy & 0.000(0.000)  & 0.955(0.008) & \textbf{0.959}(0.008) & 0.950(0.008) \\
		
		Ambiguity & 3.083(0.084)  & 1.006(0.003) & \textbf{1.000}(0.000) & \textbf{1.000}(0.000) \\
		Time(s) & \textbf{$1.067 \times 10^4$}  & $5.784 \times 10^4$ & $1.163 \times 10^4$ & $1.193 \times 10^4$ \\
		\hline
	\end{tabular}
	\label{t7}
\end{table}

In the multi-class simulation with \(p = 1000\), \(n_k = 1000\), and \(\rho = 0.8\), the results can be summarized as follows. All methods successfully control Class-wise FDR and SCW-FDR at the nominal level \(\alpha = 0.05\), but fail to control overall FDR.  The MMDCP methods achieve perfect Power, zero FLR, compact prediction sets, high accuracy, and adequate coverage, thereby substantially outperforming BCOPS. Additionally, MMDCP is computationally efficient, with runtimes far below BCOPS(rm) and only slightly higher than BCOPS(lr). Additional simulations reported in the Supplementary Material corroborate these findings, further confirming MMDCP as a robust and practically viable approach for multi-class classification and outlier detection.

\subsection{Real datasets}
In this subsection, we evaluate BCOPS (lr), BCOPS (rm), and MMDCP on the Wheat Seeds and Control Charts datasets. For all real-world datasets, 70\% of samples are used for training and 30\% for testing.

\subsubsection{Wheat Seeds Dataset}
The Wheat Seeds dataset, obtained from the UCI Machine Learning Repository, contains 199 samples from three wheat varieties: Kama, Rosa, and Canadian. Each sample is characterized by seven morphological attributes: area, perimeter, compactness, kernel length, kernel width, asymmetry coefficient, and kernel groove length.
The dataset is available at the following link: \url{https://www.kaggle.com/code/jmcaro/machine-learning-classifiers-wheat-seeds}.

\begin{table}[!htbp]
	\centering
	\renewcommand{\arraystretch}{1}
	\captionsetup{font=footnotesize}
	\caption{Multi-class label prediction and outlier detection result on Wheat Seeds dataset.}
	\footnotesize
	\begin{tabular}{>{\centering\arraybackslash}m{4cm}*{3}{>{\centering\arraybackslash}m{2.5cm}}}
		\hline
		Methods  & BCOPS(lr) & BCOPS(rm) & MMDCP\\
		\hline
		Class-wise FDR (1) &\textbf{0.000}(0.000) & \textbf{0.000}(0.000)& 0.005(0.015)  \\
		Class-wise FDR (2)  & \textbf{0.000}(0.000) & \textbf{0.000}(0.000) & 0.006(0.012)\\
		Class-wise FDR (3)  & \textbf{0.000}(0.000) & \textbf{0.000}(0.000)& 0.018(0.019) \\
		SCW-FDR  & \textbf{0.000}(0.000) & \textbf{0.000}(0.000) & 0.011(0.009)\\
		FDR  & \textbf{0.000}(0.000) & \textbf{0.000}(0.000) & 0.023(0.035)\\
		Power  & 0.000(0.000) & 0.000(0.000) & \textbf{1.000}(0.000)\\
		Accuracy  & 0.000(0.000) & 0.000(0.000)& \textbf{0.619}(0.043) \\
		Coverage & \textbf{1.000}(0.000) & \textbf{1.000}(0.000) & 0.975(0.019) \\
		Ambiguity  & 3.000(0.000) & 3.000(0.000) & \textbf{1.366}(0.051)\\
		FLR  & 0.164(0.000) & 0.164(0.000) & \textbf{0.000}(0.002)\\
		Time(s)  & 20 & 145 & \textbf{8}\\
		\hline
	\end{tabular}
	\label{tab:wheat_seed}
\end{table}

We evaluate BCOPS (lr), BCOPS (rf), and MMDCP on the Wheat Seeds dataset. As reported in Table 3, both BCOPS variants strictly control class-wise FDR, SCW-FDR, and overall FDR at zero; however, this comes at the expense of zero power and accuracy, resulting in a complete failure to assign correct class labels. In contrast, MMDCP achieves perfect power (1.0) and substantially higher accuracy, while still maintaining rigorous control of class-wise FDR, SCW-FDR, and FDR. Moreover, MMDCP yields lower ambiguity at the nominal coverage level, attains zero FLR compared with 0.164 for BCOPS, and exhibits superior computational efficiency.

\subsubsection{Control Charts dataset}
The Control Charts dataset consists of 600 synthetically generated time series, each comprising 60 observations. It encompasses six categories—Normal, Cyclic, Increasing Trend, Decreasing Trend, Upward Shift, and Downward Shift, with 100 samples per class. The dataset is publicly available at \url{https://archive.ics.uci.edu/dataset/139/synthetic+control+chart+time+series}.

\begin{table}[!htbp]
	\centering
     \renewcommand{\arraystretch}{1}
	\captionsetup{font=footnotesize}
	\caption{Multi-class label prediction and outlier detection result on Control Charts dataset.}
	\footnotesize
	\begin{tabular}{>{\centering\arraybackslash}m{4cm}*{3}{>{\centering\arraybackslash}m{2.5cm}}}
		\hline
		Methods & BCOPS(lr) & BCOPS(rm) &  MMDCP\\
		\hline
		Class-wise FDR (1) & 0.003(0.007)  & 0.005(0.007) & 0.010(0.010)  \\
		
		Class-wise FDR (2) & \textbf{0.000}(0.000)  & 0.003(0.004) & 0.009(0.012) \\
		
		Class-wise FDR (3) & \textbf{0.000}(0.000)  & 0.003(0.005) & 0.003(0.006)\\
		
		Class-wise FDR (4) & \textbf{0.002}(0.005)  & 0.006(0.006) & 0.012(0.005) \\
		
		Class-wise FDR (5) & \textbf{0.000}(0.002)  & 0.004(0.005) & 0.022(0.007)  \\
		
		Class-wise FDR (6) & \textbf{0.000}(0.000)  & 0.005(0.007) & 0.003(0.003)  \\
		
		SCW-FDR & \textbf{0.000}(0.000)  & {0.005}(0.002) & 0.010(0.002) \\
		
		FDR & \textbf{0.000}(0.000)  & {0.192}(0.199) & 0.114(0.041)  \\
		
		Power & 0.000(0.000)  & {0.214}(0.186) & \textbf{1.000}(0.000)  \\
		FLR & 0.167(0.000)  & 0.136(0.000)  & \textbf{0.000}(0.000)  \\
		
		Coverage & \textbf{0.996}(0.007)  & 0.975(0.013) & 0.945(0.011) \\
		
		Accuracy & 0.000(0.000)  & \textbf{0.507}(0.019) & \textbf{0.507}(0.254) \\
		
		Ambiguity & 5.739(0.363)  & 1.528(0.358) &  \textbf{1.464}(0.024) \\
		Time(s) & 16  & 168 & \textbf{5}  \\
		\hline
	\end{tabular}
	\label{real_tr}
\end{table}

Table 4 shows that all three methods effectively control class-wise and SCW-FDR. BCOPS(lr) achieves near-perfect FDR control, but this comes at the cost of extremely low power, resulting in almost no correct class predictions despite slightly higher coverage. In contrast, BCOPS(rf) provides a more balanced profile, offering moderate power, improved FLR, reasonable accuracy, and reduced ambiguity, although it incurs the longest runtime. MMDCP, by comparison, attains very high power with zero FLR, maintains acceptable FDR, achieves the lowest ambiguity, delivers comparable accuracy, and runs substantially faster, thereby providing the most balanced and computationally efficient performance overall.

\section{Conclusion}
In this work, we introduced the Modified Mahalanobis Distance Conformal Prediction (MMDCP), a unified framework for simultaneous outlier detection and multi-class classification. The method is supported by rigorous theoretical guarantees, including a non-asymptotic convergence rate of empirical to oracle conformal $p$-values, valid coverage, and explicit convergence rates for the adaptive prediction sets it produces. We also propose the Summarized Class-Wise FDR (SCW-FDR), a novel multiple-testing criterion that enhances error control, improves power, and increases interpretability relative to conventional FDR. Through simulation studies and real-data analyses, we demonstrate that MMDCP consistently outperforms existing methods. Future work will focus on relaxing assumptions regarding covariate magnitudes to expand applicability and developing a formal framework to rigorously analyze the method’s power properties.

\section{Competing interests}
No competing interests are declared.

\section{Author contributions statement}

Youwu Lin designed and conceived the experiments; Jingru Wu and Qi Liu conducted the experiments; Xiaoyu Qian analyzed the results; Youwu Lin and Pei Wang wrote and reviewed the manuscript.

\section{Acknowledgments}
The authors thank Professor Hansheng Wang (Peking University) for suggesting the Modified Mahalanobis Distance for outlier detection and Professor Xu Guo (Beijing Normal University) for valuable suggestions. This work was supported by the Guangxi Natural Science Foundation (2025GXNSFAA069842), the National Natural Science Foundation of China (11801105), and the Science and Technology Projects of Guangzhou (2025A04J3386).

\bibliography{reference}

\end{document}